\documentclass[12pt,english,reqno]{smfart}
\usepackage[utf8]{inputenc}
\usepackage{amsmath, amssymb, amsfonts}
\usepackage{mathrsfs}
\usepackage{geometry}
\usepackage{enumitem}
\usepackage{soul}
\usepackage{booktabs}
\usepackage{array}
\usepackage{longtable}
\usepackage{xcolor}
\usepackage{colortbl}
\usepackage{pdflscape}
\usepackage{cite}
\usepackage{caption}

\usepackage[compat=1.1.0]{tikz-feynman}
\usepackage{tikz,contour}
\usetikzlibrary{calc,decorations.markings,positioning}
\usetikzlibrary{arrows.meta, decorations.markings}

\definecolor{headerblue}{RGB}{220, 235, 250}
\definecolor{rowgray}{RGB}{245, 245, 245}

\newcolumntype{C}[1]{>{\centering\arraybackslash}p{#1}}
\newcolumntype{R}[1]{>{\centering\arraybackslash}p{#1}}

\newcommand{\F}[2]{\dfrac{#1}{#2}}
\newcommand{\Fs}[2]{{\small\dfrac{#1}{#2}}}
\newcommand{\Ff}[2]{{\footnotesize\dfrac{#1}{#2}}}
\newcommand{\Fsc}[2]{{\scriptsize\dfrac{#1}{#2}}}

\newcommand{\Real}{\Re{e}}

\setlist[itemize]{topsep=3mm, itemsep=1.8mm, parsep=1.2mm}
\newlist{conditions}{itemize}{1}
\setlist[conditions]{leftmargin=26mm, rightmargin=12mm, topsep=3mm, itemsep=1.8mm, parsep=1.2mm}
\makeatletter
\newcommand{\cond}[1]{\item[(\textsc{#1})]\protected@edef\@currentlabel{\textsc{#1}}}
\newcommand{\condconst}[2]{\item[($\text{\textsc{#1}} \mid #2$)]\protected@edef\@currentlabel{$\text{\textsc{#1}} \mid #2$}}
\makeatother

\newcommand{\undernotation}[1]{\boldsymbol{#1}}
\newcommand{\intsunset}{I_\circleddash}
\newcommand{\jintsunset}{\mathcal{I}_\circleddash}
\newcommand{\nintsunset}{\mathfrak{I}_\circleddash}
\newcommand{\persunset}{\pi_\circleddash}
\newcommand{\Msunset}{\mathcal{M}[I_\circleddash^L]}

\newcommand{\Res}[2]{\mathop{\mathrm{Res}}_{#2}\left(#1\right)}
\newcommand{\Frob}{\operatorname{Frob}}

\usepackage{xcolor}
\definecolor{bubblegum}{rgb}{0.99, 0.76, 0.8}
\definecolor{ietocean}{rgb}{0, 30, 140}
\definecolor{ietcoast}{RGB}{0, 150, 173}
\definecolor{ietlagoon}{RGB}{0, 216, 180}
\definecolor{veriforest}{RGB}{0, 80, 40}
\definecolor{revred}{RGB}{153, 0, 0}
\usepackage[]{hyperref}
\hypersetup{
	colorlinks,%
	citecolor={ietlagoon},%
	filecolor={ietocean},%
	linkcolor={veriforest},%
	urlcolor={revred},
	bookmarksnumbered=true,     
	bookmarksopen=true,         
	bookmarksopenlevel=1,       
	pdfstartview=Fit,           
	pdfpagemode=UseOutlines,
	pdfpagelayout=TwoPageRight
}

\usepackage{cleveref}
\crefname{equation}{eq.}{eqs.}
\crefname{section}{sec.}{secs.}
\Crefname{equation}{Eq.}{Eqs.}
\Crefname{section}{Sec.}{Secs.}

\newtheorem{theorem}{Theorem}[section]

\theoremstyle{definition}
\newtheorem{definition}{Definition}[section]

\newtheorem{proposition}{Proposition}[section]
\theoremstyle{remark}
\newtheorem{remark}{Remark}[section]

\begin{document}
\title{The multiloop sunset to all orders}
\author{Pierre Vanhove}
\date{\today}
\email{pierre.vanhove@ipht.fr}
 \address{
Institut de Physique Th{\'e}orique, Universit{\'e} Paris Saclay,
CNRS, CEA, F-91191 Gif-sur-Yvette, France}
\begin{abstract}
We derive exact, convergent representations of multiloop sunset Feynman 
integrals in two dimensions for arbitrary mass configurations and all loop 
orders valid  for large Euclidean momentum. The integrals are expressed as sums of symmetric polynomials in 
logarithmic mass ratios, normalized by the external momentum squared, with 
coefficients determined by analytic series expansions. For the equal-mass 
case, we establish a dimension-lowering relation expressing the $L$
loop sunset integrals in $D+2$ as the one in $D$ dimensions acted on
by a
differential operator of order $L-1$.

These representations are free of complicated transcendental functions, 
making them well-suited to both formal analysis and high-precision numerical 
evaluation. The two-dimensional results serve as boundary conditions for 
dimension-shifting relations, enabling systematic reconstruction of 
four-dimensional sunset integrals via analytic continuation to 
$D = 4 - 2\epsilon$.
\end{abstract}

\maketitle
\tableofcontents

\section{Introduction}

 The multiloop sunset graph in fig.~\ref{fig:sunset} plays a 
fundamental role in quantum field theory~\cite{Kallen:1955fb,
Broadhurst:1993mw,Caffo:1998du,Groote:2005ay} 
and it appears in a 
broad range of precision calculations. These include two-loop mixed 
QCD--electroweak master integrals for Higgs boson production via gluon fusion 
involving virtual top quarks~\cite{Marzucca:2025dnh}, analytic two-loop QCD 
amplitudes for diphoton production mediated by heavy-quark 
loops~\cite{Becchetti:2025rrz}, and the three-loop QED photon 
self-energy~\cite{Forner:2024ojj}, or vacuum polarisation in chiral
perturbation theory~\cite{Lellouch:2025rnz}.
Their evaluation at multiple loops is notoriously challenging due to 
the presence of elliptic integrals and transcendental structures beyond 
multiple polylogarithms~\cite{Laporta:2004rb,Bloch:2014qca,
Adams:2013nia,Remiddi:2013joa,Remiddi:2016gno,Bloch:2016izu,
Broedel:2017siw,Broedel:2018iwv,Primo:2017ipr,
Muller-Stach:2011qkg,Kniehl:2005bc}.

The expression for the multiloop sunset  integral in $D=2$ dimensions is
\begin{equation}\label{e:sunsetdef}
  \intsunset^{(L)}(p^2,\undernotation m^2):={1\over \pi^L} \int
  { d^2\ell_1\cdots d^2\ell_L\over
    (\ell_1^2-m_1^2+i\varepsilon) \cdots (\ell_L^2-m_L^2+i\varepsilon) ((\ell_1+\cdots+\ell_L-p)^2-m_{L+1}^2+i\varepsilon)}  
\end{equation}
where $\undernotation m^2 := (m_1^2, \ldots, m_{L+1}^2)$ denotes the vector of squared masses.
This integral is both ultraviolet and infrared finite. We work in the large
Euclidean region $t:=-p^2\gg1$ .

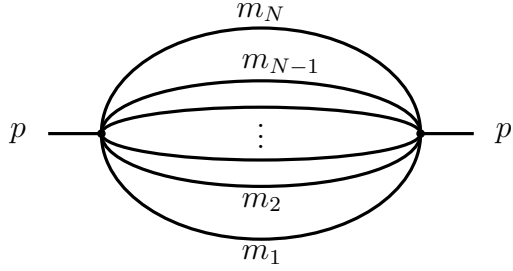
\begin{figure}[t]
                        \centering
                        \begin{tikzpicture}[scale=0.7]
                                \draw[very thick] (0,0) ellipse (3cm and 2cm);
                                \draw[very thick] (0,0) ellipse (3cm and 1.cm);
                                \draw [very thick](0,0) ellipse (3cm and 0.5cm);
                                \draw [very thick] (-3,0)--(-4,0);
                                \draw [very thick] (3,0)--(4,0);
                \filldraw [black] (3,0) circle (2pt);
                \filldraw [black] (-3,0) circle (2pt);
                                \node[text width=0.5cm, text centered ] at (0,0.1) {$\vdots$};
                                \node [text width=0.5cm, text centered ] at (0,-2.3){$m_1$};
                                \node [text width=0.5cm, text centered ] at (0,-1.3){$m_2$};
                                \node [text width=0.5cm, text centered ] at (0,1.3){$m_{N-1}$};
                                \node [text width=1.5cm, text centered
                                ] at (0,2.3){$m_{N}$};
                                 \node [text width=0.5cm, text
                                 centered,  right] at (4,0){$p$};
                                  \node [text width=0.5cm, text centered,left ] at (-4,0){$p$};
                        \end{tikzpicture}
                        \caption{Multiloop sunset graph.}
                        \label{fig:sunset}
                \end{figure}

\medskip
                
\noindent{\bf Our main results are:}

\begin{itemize}
  \item[$\bullet$] An exact, convergent expansion valid for generic
mass configurations for the multiloop sunset integral in two
dimensions, stated in theorem~\ref{thm:generic}, 
  \begin{multline}\label{e:IsunsetResultinto}
\intsunset^{(L)}(p^2,\undernotation m^2)=-{1\over p^2  } \sum_{(r_1,\dots,r_{L+1})\in\mathbb
  N^{L+1}}\, {(r_1+\cdots+r_{L+1})!^2\over r_1!^2\cdots r_{L+1}!^2}
\prod_{i=1}^{L+1}\left(-{m_i^2\over p^2}\right)^{r_i}\cr
\times
\sum_{k=1}^{L+1} c_k(r_1+\cdots +r_{L+1}) P_{L+1}^{L+1-k}\left(\ell_1(r_1),\cdots,\ell_{L+1}(r_{L+1})\right)
\end{multline}
where the coefficients $c_r(k)$ are given in equation~\eqref{e:Cdef}, and $P_{L+1}^{r}(\ell_1,\dots,\ell_{L+1})$ are the symmetric
polynomials of order $r$, defined in equation~\eqref{eq:polydef}, in the variables
\begin{equation}
    \ell_{i}(r):=\log\left(-\frac{m_{i}^{2}}{p^2}\right)-2\sum_{n=1}^r\frac{1}{n}.
  \end{equation}
From now we will drop the $+i\epsilon$ prescription on the masses and
choose the principal value determination of the logarithm so that
$\log((m^2-i\varepsilon)/(-p^2))= \log|m^2/p^2|+ i \arg((m^2-i\varepsilon)/(-p^2))$.
This series converges absolutely for $|p^2| > \left(m_1+\cdots+m_{L+1}\right)^2$ and
represents the exact value of the integral, not merely an asymptotic
expansion above the normal threshold. 

\item[$\bullet$] In the all-equal-mass case we derive in theorem~\ref{thm:equalmass},
the alternative expression 
\begin{multline}\label{e:equalmassintro}
 \intsunset^{(L)}(p^2,\undernotation 1) =
 {d  R^{(L)}(p^2)\over dp^2}\,\left(- (L+1) \left(-
    R^{(L)}(p^2)\right)^{L} + \sum_{n=1}^{\infty} e^{n\,R^{(L)}(p^2)}
  \sum_{s=0}^{L-1} d_{s}^{L}(n) \left(R^{(L)}(p^2)\right)^{s}
\right)\cr
+ \sum_{r=0}^{L-1}
    \alpha_r^{(L)}\,\Frob_r^{(L)}(p^2),
\end{multline}
with rational coefficients $d_s^L(n)$, $\alpha_r^{(L)}$ are
polynomials in the odd zeta values of weight $L-r$,
$\Frob_r^{(L)}(p^2)$ are the elements of the Frobenius basis defined
in equation~\eqref{eq:Frob-def},    and 
\begin{equation}
  R^{(L)}(p^2)
  \;=\;i \pi
  -\log(p^2)
  \;+\;
  \sum_{\substack{(r_1,\ldots,r_{L+1})\in\mathbb{N}^{L+1}\\
                  r_1+\cdots+r_{L+1}>0}}
  \left(\frac{(r_1+\cdots+r_{L+1})!}{r_1!\cdots r_{L+1}!}\right)^{\!2}
  \frac{(p^2)^{-(r_1+\cdots+r_{L+1})}}{r_1+\cdots+r_{L+1}},
\end{equation}

\item[$\bullet$] In
  section~\ref{sec:finite4d}, we derive a dimension-lowering relation,
  where the
  all-equal-mass sunset integral in $D+2$ dimensions is obtained from 
  the sunset integral in $D$ dimensions by the action of a
  differential operator. We apply this relation to $\epsilon$ expansion of
  the all-equal-mass sunset in $D=4-2\epsilon$, and express the finite part in four dimensions from
derivative of the finite part in two dimensions. Therefore, the all-equal-mass sunset integral in $D=2-2\epsilon$ serves as boundary conditions for dimension-shifting relations that extend to $D = 4-2\epsilon$, enabling systematic computation of four-dimensional sunset integrals.

\end{itemize}

These results yield expressions for the multiloop sunset integrals
that avoid complicated transcendental functions and are readily
amenable to both formal and numerical analysis.

\bigskip

The result in equation~\eqref{e:IsunsetResultinto} is expected from  Weinberg's analysis of the
asymptotic behaviour of Feynman integrals in the large euclidean regime~\cite{Weinberg:1959nj}. 
If  $I_\Gamma(\lambda)$ is a Feynman integral where a subset of the
kinematic invariants are scaled by $\lambda$, the theorem in~\cite{Bergere:1977wz} states
that the asymptotic behaviour of $I_\Gamma(\lambda)$, for large
$\lambda$ and  Euclidean momentum, 
takes the form
\begin{equation}
\label{e:asymptotic}
  I_\Gamma(\lambda)\simeq \sum_{r=0}^{r_{\textrm{max}}}
  (\log(\lambda))^r \sum_{s\geq s_0} {g_{r,s}\over \lambda^s}
\end{equation}
where $r_{\textrm{max}}$ and $s_0$ are finite positive integers, and
$g_{r,s}$ depend on the internal masses and kinematic variables, but not on $\lambda$.
In the case of the multiloop massive sunset integral in two dimensions,
because the integral is convergent we have  $r_{\textrm{max}}=L$
the loop order, and the
expression  in equation~\eqref{e:asymptotic} is not an asymptotic expansion but
an \emph{exact} expression.
In this work, we derive the
expression in equation~\eqref{e:IsunsetResultinto} by determining all the coefficients in the expansion  to all
loop orders and all mass configurations using the Mellin transform method in
section~\ref{sec:Mellin}.
\medskip

The all-equal-mass result in equation~\eqref{e:equalmassintro} was proved for the case of the generic-mass two-loop
sunset graph in~\cite{Bloch:2016izu}.
At higher-loop order, our result establishes an
exact version of the asymptotic expansion observed numerically up to four loops
in~\cite{Klemm:2019dbm,Bonisch:2020qmm} and the leading term proved by~\cite{Kerr:2022,Iritani:2022}.

\bigskip

This paper has the following structure. In section~\ref{sec:Mellin} the
representation in equation~\eqref{e:IsunsetResultinto} is derived. In
section~\ref{sec:diffeq} we make some comment on the inhomogeneous
term of the minimal Picard-Fuchs operator with respect to $p^2$ satisfied by the sunset integral in
two dimensions. In section~\ref{sec:equalmass} we derive the
expression in equation~\eqref{e:equalmassintro} for the all-equal-mass case. In
section~\ref{sec:finite4d} we derive a dimension-lowering formula
expressing the all-equal-mass sunset integral in $D+2$ dimensions in
terms of the same 
same integral in
$D$ dimensions.
The Appendices~\ref{sec:3loop1mass}
and~\ref{sec:4loop1mass} contain the coefficients in
 equation~\eqref{e:equalmassintro} at three and four loops, respectively.

\bigskip

The representations of the multiloop sunset integral in two dimensions
in equation~\eqref{e:IsunsetResultinto} and equation~\eqref{e:equalmassintro} implemented in
\texttt{SageMath} codes~\cite{sage} are available in the
repository~\cite{repoAllLoopSunset}, as is a  \texttt{Maple} implementation of the dimension-lowering formula.

\section*{Acknowledgements}
We thank Leonardo de la Cruz for discussions and comments, and and
Mattias Sj\"o for a thorough reading of the manuscript and discussions.
The work of PV was funded by the Agence Nationale de la Recherche
(ANR) under the grant Observables (ANR-24-CE31-7996).
\section{Mellin representation}\label{sec:Mellin}

The Mellin transform method provides a systematic approach to evaluating
Feynman integrals by transforming them into contour integrals in the
complex plane. The key advantage of this representation is that it
exposes the analytic structure of the integral and allows us to
extract series expansions by residue calculus. The method used here
follows the techniques developed
in~\cite{Bergere:1973fq,Bergere:1977wz} and applied successfully to
various Feynman integrals as presented in textbooks~\cite{Smirnov:2004ym,Weinzierl:2022eaz}.  The sunset case is particularly elegant due to its complete convergence.

We start with the parametric representation of the
integral in equation~\eqref{e:sunsetdef} (see e.g.~\cite{Vanhove:2014wqa} for a derivation)
\begin{equation}\label{eq:parametric}
\intsunset^{(L)}(p^{2},\undernotation m^2)=\int_{\mathbb{R}_{+}^{L+1}}e^{-\sum_{i=1}^{L+1}m_{i}^{2}x_{i}+\frac{p^2}{\sum_{i=1}^{L+1}x_{i}^{-1}}}\frac{1}{\sum_{i=1}^{L+1}x_{i}^{-1}}\frac{dx_{1}\cdots dx_{L+1}}{x_{1}\cdots x_{L+1}}.
\end{equation}

\subsection{Mellin transform}

Setting $t=-p^2$, the Mellin transform of $\intsunset^{(L)}(t)$ is defined as
\begin{equation}\label{eq:mellindef}
\Msunset(\zeta):=\int_{0}^{\infty}\intsunset^{L}(-t,\undernotation m^2)\,t^{-1-\zeta}\,dt.
\end{equation}
Using equation~\eqref{eq:parametric} and the standard integral
$\int_{0}^{\infty}e^{-t/u}t^{-1-\zeta}dt=\Gamma(-\zeta)u^{-\zeta}$
valid for $\Real (\zeta)<0$, we obtain
\begin{equation}\label{eq:mellinintermediate}
\Msunset(\zeta)=\frac{\Gamma(-\zeta)}{\Gamma(1+\zeta)}\int_{\mathbb{R}_{+}^{L+1}}e^{-\sum_{i=1}^{L+1}m_{i}^{2}x_{i}}\left(\sum_{i=1}^{L+1}x_{i}^{-1}\right)^{-1-\zeta}\frac{dx_{1}\cdots dx_{L+1}}{x_{1}\cdots x_{L+1}}.
\end{equation}

To evaluate the remaining integral, we employ the Mellin
transformation given in~\cite[Eq.~8.4.2]{Gradshteyn:1965} 
\begin{equation}\label{eq:MBformula}
\frac{\Gamma(x)}{(A+B)^{x}}=\int_{-\infty}^{\infty}\frac{\Gamma(-\sigma-iy)}{A^{-\sigma-iy}}\frac{\Gamma(\sigma+iy+x)}{B^{\sigma+iy+x}}\frac{dy}{2\pi i},
\end{equation}
valid for $\Real(x),\Real(A),\Real(B)>0$ and $-\Real(x)<\sigma<0$.
Iterating this formula $L$ times for the factor $(\sum_{i=1}^{L+1}x_{i}^{-1})^{-1-\zeta}$ yields
\begin{equation}\label{eq:iteratedMB}
\left(\sum_{i=1}^{L+1}x_{i}^{-1}\right)^{-1-\zeta}=\frac{1}{\Gamma(1+\zeta)}\frac{1}{(2i\pi)^{L}}\int_{(i\mathbb{R})^{L}}\Gamma(1+\zeta+z_{1}+\cdots+z_{L})x_{L+1}^{1+\zeta}\prod_{i=1}^{L}\Gamma(-z_{i})\left(\frac{x_{L+1}}{x_{i}}\right)^{z_{i}}dz_{i}.
\end{equation}

Using $\int_{0}^{\infty}e^{-m^{2}x}x^{\alpha}\frac{dx}{x}=\frac{\Gamma(\alpha)}{(m^{2})^{\alpha}}$ for $\alpha\notin\mathbb{Z}_{\leq0}$, the parametric integrals evaluate to $\Gamma$-functions, giving the compact representation
\begin{equation}\label{eq:mellinfinal}
\Msunset(\zeta)=\frac{\Gamma(-\zeta)}{\Gamma(1+\zeta)}\frac{1}{(2i\pi)^{L}}\int_{(i\mathbb{R})^{L}}\frac{\Gamma(1+\zeta+z_{1}+\cdots+z_{L})^{2}}{m_{L+1}^{2(1+\zeta)}}\prod_{i=1}^{L}\left(\frac{m_{i}^{2}}{m_{L+1}^{2}}\right)^{z_{i}}\Gamma(-z_{i})^{2}dz_{i}.
\end{equation}

\subsection{Inverse Mellin transform and residue computation}

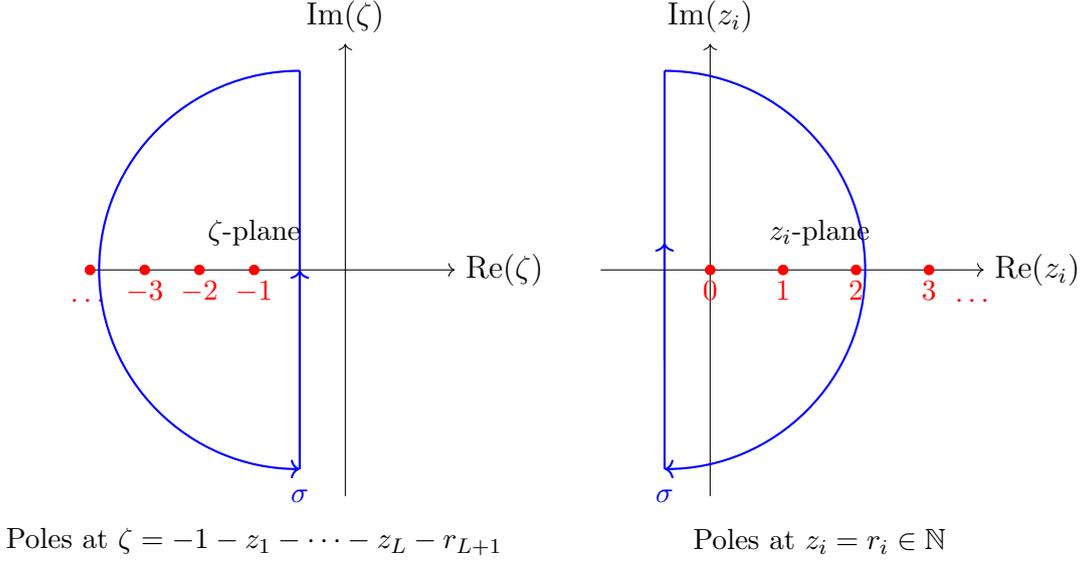
\begin{figure}[t]
\centering
\begin{tikzpicture}[scale=1.2]

\begin{scope}[xshift=-1cm]
  \draw[->] (-2.8,0) -- (1.2,0) node[right] {$\mathrm{Re}(\zeta)$};
  \draw[->] (0,-2.5) -- (0,2.5) node[above] {$\mathrm{Im}(\zeta)$};

  \foreach \n in {0,1,2,3}{
    \pgfmathsetmacro{\xpos}{-1-\n*0.6}
    \filldraw[red] (\xpos,0) circle (1.5pt);
  }
  \node[red, below] at (-1,0) {\small $-1$};
  \node[red, below] at (-1.6,0) {\small $-2$};
  \node[red, below] at (-2.2,0) {\small $-3$};
  \node[red] at (-2.8,-0.35) {\small $\cdots$};

  \draw[blue, thick] (-.5,-2.2) -- (-.5,2.2);
  \node[blue] at (-0.5,-2.5) {\small $\sigma$};

  \draw[blue, thick, ->] (-0.5,2.2) arc (90:270:2.2 and 2.2);
  \draw[blue, thick, ->] (-0.5,-2.2) -- (-0.5,0.0);

  \node at (-1, -3.0) {\small Poles at $\zeta = -1-z_1-\cdots-z_L - r_{L+1}$};
  \node[above] at (-1,0.15) {\small $\zeta$-plane};
\end{scope}

\begin{scope}[xshift=3cm]
  \draw[->] (-1.2,0) -- (3.0,0) node[right] {$\mathrm{Re}(z_i)$};
  \draw[->] (0,-2.5) -- (0,2.5) node[above] {$\mathrm{Im}(z_i)$};

  \foreach \n in {0,1,2,3}{
    \filldraw[red] (\n*0.8,0) circle (1.5pt);
  }
  \node[red, below] at (0,0) {\small $0$};
  \node[red, below] at (0.8,0) {\small $1$};
  \node[red, below] at (1.6,0) {\small $2$};
  \node[red, below] at (2.4,0) {\small $3$};
  \node[red] at (2.9,-0.35) {\small $\cdots$};

  \draw[blue, thick] (-0.5,-2.2) -- (-0.5,2.2);
  \node[blue] at (-0.5,-2.5) {\small $\sigma$};

  \draw[blue, thick, ->] (-0.5,2.2) arc (90:-90:2.2 and 2.2);
  \draw[blue, thick, ->] (-0.5,-1.0) -- (-0.5,0.3);

  \node at (1.2, -3.0) {\small Poles at $z_i = r_i \in \mathbb{N}$};
  \node[above] at (1.2,0.15) {\small $z_i$-plane};
\end{scope}

\end{tikzpicture}
\caption{Contour integration strategy for extracting the series expansion. \textbf{Left panel:} In the $\zeta$-plane, we close the contour to the left (toward $\Real(\zeta) \to -\infty$), capturing poles at $\zeta = -1-z_1-\cdots-z_L-r_{L+1}$ for all $r_{L+1} \in \mathbb{N}$. Each pole contributes a term in the final series. \textbf{Right panel:} In each $z_i$-plane, we close the contour to the right (toward $\Real(z_i) \to +\infty$), capturing poles at $z_i = r_i \in \mathbb{N}$. The $(L+1)$-dimensional residue computation yields the coefficients of the expansion. The contours can be closed because the $\Gamma$-functions ensure exponential decay in the appropriate half-planes.}
\label{fig:contour}
\end{figure}

Applying the inverse Mellin transform
\begin{equation}
\intsunset^{L}(t)=\frac{1}{2\pi i}\int_{i\mathbb{R}}\Msunset(\zeta)\,t^{\zeta}\,d\zeta,
\end{equation}
we obtain the $(L+1)$-fold contour integral representation
\begin{multline}\label{eq:multicontour}
\intsunset^{L}(p^2,\undernotation m^2)=-\frac{1}{p^2(2\pi i)^{L+1}}\int_{(i\mathbb{R})^{L+1}}\frac{\Gamma(-\zeta)\Gamma(1+\zeta+z_{1}+\cdots+z_{L})^{2}}{\Gamma(1+\zeta)}\\
\times\prod_{i=1}^{L}\left(\frac{m_{i}^{2}}{m_{L+1}^{2}}\right)^{z_{i}}\Gamma(-z_{i})^{2}\,dz_{i}\left(-\frac{p^2}{m_{L+1}^{2}}\right)^{1+\zeta}d\zeta.
\end{multline}

We evaluate this by closing contours and computing residues as
indicated in figure~\ref{fig:contour}:
\begin{itemize}
    \item Close the $\zeta$-contour to the \textbf{left}, picking poles at $\zeta=-1-z_{1}-\cdots-z_{L}-r_{L+1}$ for $r_{L+1}\in\mathbb{N}$;
    \item Close the $z_i$-contours to the \textbf{right}, picking poles at $z_{i}=r_{i}\in\mathbb{N}$ for $i=1,\ldots,L$.
\end{itemize}

The residue formula for a pole of order $k$ at $z = a$ is
\begin{equation}
\Res{f(z)}{z=a} = \frac{1}{(k-1)!} \lim_{z \to a} \frac{d^{k-1}}{dz^{k-1}} [(z-a)^k f(z)]
\end{equation}
which applies to the case of double poles gives
\begin{equation}\label{eq:residueformula}
\Res{\Gamma(-z)^{2}f(z)}{z=n}=\lim_{z\to n}\frac{d}{dz}\frac{f(z)}{\Gamma(1+z)^{2}},
\end{equation}
 leading to the single derivative per variable $z_i$.

Applying this to all $L+1$ variables yields the exact expansion:
\begin{multline}\label{eq:mellinres}
\intsunset^{L}(p^2,\undernotation m^2)=-\frac{1}{p^2}\sum_{(r_{1},\ldots,r_{L+1})\in\mathbb{N}^{L+1}}\lim_{(z_{1},\ldots,z_{L+1})\to(r_{1},\ldots,r_{L+1})}\\
\times\frac{\partial^{L+1}}{\partial z_{1}\cdots\partial z_{L+1}}\left[\frac{\Gamma(1+z_{1}+\cdots+z_{L+1})}{\Gamma(-z_{1}-\cdots-z_{L+1})}\prod_{i=1}^{L+1}\frac{1}{\Gamma(1+z_{i})^{2}}\left(-\frac{m_{i}^{2}}{p^2}\right)^{z_{i}}\right].
\end{multline}
The product structure then generates the symmetric polynomials through the multinomial theorem.

\medskip

We can now state  the main result:\footnote{An
  implementation is  given in
  the repository~\cite{repoAllLoopSunset}, where various cases are
  worked out.}
\begin{theorem}\label{thm:generic}
  The multiloop sunset integral has the exact expansion in equation~\eqref{e:IsunsetResult},
   where we capture the dependence on logarithms of $m_i^2/p^2$ via the logarithmic factors $\ell_i(r)$ defined in equation~\eqref{eq:elldef} and the symmetric polynomials $P^k_{L+1}$, whose generating function is~\eqref{eq:polydef}
\begin{multline}\label{e:IsunsetResult}
\intsunset^{(L)}(p^2,\undernotation m^2)=-{1\over p^2  } \sum_{(r_1,\dots,r_{L+1})\in\mathbb
  N^{L+1}}\, {(r_1+\cdots+r_{L+1})!^2\over r_1!^2\cdots r_{L+1}!^2}
\prod_{i=1}^{L+1}\left(-{m_i^2\over p^2}\right)^{r_i}\cr
\times
\sum_{k=1}^{L+1} c_k(r_1+\cdots +r_{L+1}) P_{L+1}^{L+1-k}\left(\ell_1(r_1),\cdots,\ell_{L+1}(r_{L+1})\right),
\end{multline}
for $n\in\mathbb N$, $c_0(n)=0$ and for $r\geq1$  
  \begin{multline}\label{e:Cdef}
 c_k(m):=  
    (-1)^{m+1}\sum_{n=0}^k \binom{k}{n}\sum_{a=0}^{n}
    \binom{n}{a} (-\pi^2)^{k-n-1\over2} \left(\frac{1 - (-1)^{k-n}}{2}\right)\cr
 \times   Y_a\left(\Psi^{(0)}(1+m),\dots,\Psi^{(a-1)}(1+m)\right)
    Y_{n-a}\left(\Psi^{(0)}(1+m),\dots,\Psi^{(n-a-1)}(1+m)\right)
  \end{multline}
where
\begin{equation}\label{e:psi-def}
  \Psi^{(k)}(m):=
  \begin{cases}
    \sum_{n=1}^{m-1} {1\over n}& \textrm{for}~k=0\\[2ex]
    (-1)^{k+1} k! \left(\zeta(k+1)-\sum_{n=1}^{m-1} {1\over n^{k+1}}\right)& \textrm{for}~k\geq1\,.
  \end{cases}
\end{equation}
  and $Y_m(x_1,\dots,x_m)$ is the complete Bell polynomial
\begin{equation}
  Y_m(x_1,\dots,x_m):=\sum_{l_1,\dots,l_m\geq0\atop
    l_1+2l_2+\cdots+ml_m=m}  {m!\over l_1!\cdots l_m!} \prod_{j=1}^m
  \left(x_j\over j!\right)^{l_j}  
\end{equation}
and we have introduced the  logarithmic factors
\begin{equation}\label{eq:elldef}
  \ell_{i}(r):=\log\left(-\frac{m_{i}^{2}}{p^2}\right)-2\sum_{n=1}^r{1\over
    n}\,,
\end{equation}
and  the symmetric polynomials $P_{L+1}^{k}$ via the generating function
\begin{equation}\label{eq:polydef}
\prod_{i=1}^{L+1}(1+tx_{i})=\sum_{k=0}^{L+1}P_{L+1}^{k}(x_{1},\ldots,x_{L+1})\,t^{k}\,.
\end{equation}
\end{theorem}

\begin{proof} 
We first note that
\begin{multline}
    \frac{\partial^{L+1}}{\partial z_{1}\cdots\partial
      z_{L+1}}\left[\frac{\Gamma(1+z_{1}+\cdots+z_{L+1})}{\Gamma(-z_{1}-\cdots-z_{L+1})}\prod_{i=1}^{L+1}\frac{1}{\Gamma(1+z_{i})^{2}}\left(-\frac{m_{i}^{2}}{p^2}\right)^{z_{i}}\right]\cr=
   \left[ \sum_{r=0}^{L+1} {\partial^r\over \partial z_1^{r}}
     \frac{\Gamma(1+z_{1}+\cdots+z_{L+1})}{\Gamma(-z_{1}-\cdots-z_{L+1})}
   P_{L+1}^{L+1-r}\left(\ell_1(z_1),\cdots,\ell_{L+1}(z_{L+1}) \right)\right]\prod_{i=1}^{L+1}\frac{1}{\Gamma(1+z_{i})^{2}}\left(-\frac{m_{i}^{2}}{p^2}\right)^{z_{i}}
\end{multline}
with $ \ell_i(z_i)$ defined in equation~\eqref{eq:elldef}.  
We use that for $n\in\mathbb N$ and $k\in \mathbb N$
\begin{equation}\label{e:ratio-gamma}
  \lim_{z\to0}\, {d^k \over dz^k}
    {\Gamma(1+n+z)\over \Gamma(-n+z)} =n!^2c_k(n)
  \end{equation}
  with $c_0(n)=0$ because of the vanishing  from the $\Gamma$-function in
  the denominator, and for $k\geq1$, $c_k(n)$ is given by
  equation~\eqref{e:Cdef} by a direct use of Leibniz rules for
  derivatives and Fa\`a di Bruno's formula for the multiple derivative
  of composition of functions~\cite{comtet_advanced_combinatorics}
  \begin{equation}
    {d^m\over dx^m} g(f(x))=    \sum_{l_1,\dots,l_m\geq0\atop l_1+2l_2+\cdots+ml_m=m} {m!\over l_1!\cdots
        l_m!} g^{(l_1+\cdots+l_m)}(f(x))\, \prod_{k=1}^m \left(f^{(k)}(x)\over k!\right)^{l_k}
    \end{equation}
    with $g^{(k)}(x):=d^kg(x)/dx^k$.
   
Furthermore for $n\in\mathbb N$, we make use of 
\begin{equation}
  \lim_{z\to n} {d\over dz} {X^z\over \Gamma(1+z)^2}={X^n\over n!^2} \left(\log
    X+2\gamma_E-2\sum_{m=1}^n {1\over m}\right)\,.
\end{equation}

The  Euler-Gamma, $\gamma_E$, does not contribute to the
final expression, because its dependence on the $z_i$ cancels in  the ratio of $\Gamma$-functions
\begin{equation}
\Gamma(1+z_1+\cdots+z_{L+1})\over \Gamma(-z_1-\cdots-z_{L+1})\prod_{i=1}^{L+1}\Gamma(1+z_i)^2
\end{equation}
as can be seen using the product representation of the
$\Gamma$-function
\begin{equation}
  \Gamma(z)={e^{-\gamma_E z}\over z}  \prod_{n=1}^\infty
  \left(1+{z\over n}\right)^{-1}e^{z\over n}\,.
\end{equation}

Consequently, there is no Euler-Gamma dependence in $\Psi^{(0)}(1+m)$ in
equation~\eqref{e:psi-def}  nor in $\ell_i(r)$ in equation~\eqref{eq:elldef}.

\medskip

Collecting everything leads to the result in equation~\eqref{e:IsunsetResult}
and completes the proof.
\end{proof}

\medskip

\begin{remark}[Asymptotic of Feynman integrals]
  This series converges absolutely for $|p^2| > \left(m_1+\cdots+m_{L+1}\right)^2$ and represents the exact value of the integral, not merely an asymptotic expansion.
This result is the convergent case of the general asymptotic theorem
of Weinberg analysis in~\cite{Weinberg:1959nj} and the theorem of
Bergère, de Calan, and Malbouisson~\cite{Bergere:1977wz}: for the case
of the sunset integrals, the asymptotic expansion becomes exact with
$r_{\max}=L$ equal to the loop order.
\end{remark}

\section{Structure of the Picard--Fuchs differential equation}
\label{sec:diffeq}

The exact expansion derived in Section~\ref{sec:Mellin} is consistent
with the differential equation satisfied by the sunset
integral. In this section we make the relationship between the two
perspectives explicit. We show that the source term of the Picard--Fuchs
equation has a precise transcendental structure that mirrors the logarithmic
architecture of the expansion in Theorem~\ref{thm:generic}, and we identify
the special features of the all-equal-mass case that make it analytically
tractable at all loop orders. This analysis motivates and sets up the
construction of Section~\ref{sec:equalmass}.

\subsection{The inhomogeneous Picard--Fuchs equation}
\label{subsec:PF-inhom}

For generic masses the multiloop sunset integral satisfies the inhomogeneous
differential equation
\begin{equation}
  \mathscr{L}_L\, \intsunset^{(L)}(p^2,\undernotation m^2)
  \;=\;
  \mathscr  S_L(p^2,\undernotation m^2),
  \label{eq:PFeq}
\end{equation}
where the differential operator
\begin{equation}
  \mathscr{L}_L
  \;:=\;
  \sum_{r=0}^{O(L)} q_r(p^2,\undernotation m^2)
  \left(\frac{d}{dp^2}\right)^r
  \label{eq:PFop}
\end{equation}
is a Calabi--Yau differential operator~\cite{Bogner,Bloch:2016izu,Bloch:2014qca,Klemm:2019dbm,Bonisch:2020qmm,Bonisch:2021yfw,Lairez:2022zkj,Duhr:2022dxb,Maggio:2025jel}
of order
\begin{equation}
  O(L)
  \;=\;
  2^{L+1}-\binom{L+2}{\left\lfloor L+2\over 2\right\rfloor},
  \label{eq:PForder}\,.
\end{equation}
The coefficients $q_r(p^2,\undernotation m^2)$ are polynomials in $p^2$
and the internal masses $\undernotation m^2$. The real singularities
of the differential
operator are the roots of the leading coefficients
$q_{O(L)}(p^2,\undernotation m^2)$  
which are located at the roots of the discriminant locus
\begin{equation}
  \Delta_L(p^2,\undernotation m^2)
  \;=\;
  p^2
  \prod_{\epsilon\in\{-1,+1\}^{L+1}}
  \!\!\left(p^2 - \Bigl(\sum_{i=1}^{L+1}\epsilon_i m_i\Bigr)^2\right).
  \label{eq:discriminant}
\end{equation}
This discriminant is precisely the singular locus of the \emph{sunset
hypersurface}
\begin{equation}
  \left(\sum_{i=1}^{L+1}\frac{1}{x_i}\right)
  \!\left(\sum_{i=1}^{L+1}m_i^2\,x_i\right)
  \;=\; p^2,
  \label{eq:sunset-hyp}
\end{equation}
which defines the underlying Calabi--Yau variety. The reduction from the
naive dimension $2^L$ to $O(L)$ in \eqref{eq:PForder} reflects algebraic
relations among the periods in integer dimension~\cite{Tancredi:2015pta,Lairez:2022zkj,delaCruz:2024xit,Lellouch:2025rnz}.

\subsection{The holomorphic period}
\label{subsec:hol-period}

The homogeneous equation $\mathscr{L}_L\, \persunset=0$ admits the
analytic solution near $p^2=\infty$ given by the period integral~\cite{Vanhove:2018mto}
\begin{equation}
 \persunset^{(L)}(p^2,\undernotation m^2)
  \;:=\;
  \int_{|x_1|=\cdots=|x_{L+1}|=1}
  \frac{dx_1\cdots dx_{L+1}}
  {p^2
   - \bigl(\sum_{i=1}^{L+1}m_i^2 x_i\bigr)
     \bigl(\sum_{i=1}^{L+1}x_i^{-1}\bigr)},
  \label{eq:hol-period}
\end{equation}
which arises as the \emph{maximal cut}—the highest-transcendental-weight
part—of the Feynman
integral~\cite{Primo:2016ebd,Primo:2017ipr,Frellesvig:2017aai,Harley:2017qut}. By
definition it satisfies $\mathscr{L}_L\, \persunset=0$ and has the series expansion in generalised Ap\'ery
numbers~\cite{Verrill2004}
\begin{equation}
 \persunset^{(L)}(p^2,\undernotation m^2)
  \;=\;
  \frac{1}{p^2}
  \sum_{(r_1,\ldots,r_{L+1})\in\mathbb{N}^{L+1}}
  \left(\frac{(r_1+\cdots+r_{L+1})!}
             {r_1!\cdots r_{L+1}!}\right)^{\!2}
  \prod_{i=1}^{L+1}\left(\frac{m_i^2}{p^2}\right)^{\!r_i}.
  \label{eq:period-series}
\end{equation}

\subsection{Structure of the source term}
\label{subsec:source}

Given the exact expansion of Theorem~\ref{thm:generic}, we can extract
the precise form of the source term by substituting equation~\eqref{e:IsunsetResult}
into equation~\eqref{eq:PFeq}. The key identity is
\begin{equation}
 \frac{d^r}{(dp^2)^r}
  \bigl(\log p^2\bigr)^{L}
  \;=\;
  \frac{L!}{(L-r)!}
  \frac{\bigl(\log p^2\bigr)^{L-r}}{(p^2)^r}
  \;+\;
  \text{lower powers of }\log p^2,
  \label{eq:logderiv}
\end{equation}
which shows that applying $\mathscr{L}_L$ to a term of the form
$\persunset\cdot(\log p^2)^L$ produces logarithmic powers up to
$(\log p^2)^{L-1}$ on the right-hand side.

\medskip

\begin{proposition}[Structure of the source term]
\label{prop:source}
The source term $\mathscr S_L(p^2,\undernotation m^2)$ in equation~\eqref{eq:PFeq} takes
the form
\begin{equation}
 \mathscr S_L(p^2,\undernotation m^2)
  \;=\;
  P(p^2,\undernotation m^2)
  \;+\;
  \sum_{a=0}^{L-1}
  s_a(p^2,\undernotation m^2)\,(\log p^2)^a,
  \label{eq:source-generic}
\end{equation}
where $P$ and $s_a$ are polynomials in $p^2$. For generic masses, the
coefficients $s_a$ contain logarithms $\log m_i^2$, reflecting the explicit
mass dependence in the variables $\ell_i(r)$ defined in equation~\eqref{eq:elldef}.

In the all-equal-mass case $m_1=\cdots=m_{L+1}=1$, all mass logarithms
vanish identically. The differential operator reduces to order $L$, and the
source term simplifies to the constant
\begin{equation}
 \mathscr S_L(p^2,\mathbf{1}) \;=\; -(L+1)!\,.
  \label{eq:source-equal}
\end{equation}
\end{proposition}

\begin{proof}
In the all-equal-mass case, the variables $\ell_i(r)$ in equation~\eqref{eq:elldef}
are all equal to $\log(1/p^2)-2\sum_{n=1}^r n^{-1}$, so that no individual mass
logarithm $\log m_i^2$ appears. The source term is then polynomial in $p^2$
by construction; that it equals the specific constant $-(L+1)!$ follows from
the boundary condition at the threshold $p^2=(L+1)^2$, and is verified by
explicit computation of the differential equations at two through five loops~\cite{Bloch:2016izu,Remiddi:2016gno,Muller-Stach:2011qkg,Lairez:2022zkj} 
available in the repository~\cite{repoPF}.
\end{proof}

\begin{remark}
The cancellation of all mass logarithms and the reduction to a constant
source term in the equal-mass case is not merely a computational
simplification: it reflects the higher symmetry of the equal-mass sunset
variety, which acquires the full symmetric group $\mathfrak{S}_{L+1}$
as an automorphism group. This symmetry constrains the differential equation
drastically and is what ultimately allows the exact all-order treatment of
Section~\ref{sec:equalmass}.
\end{remark}

\subsection{The one-loop case as a prototype}
\label{subsec:1loop}

The one-loop case is fully explicit and illustrates all the structures that
will appear at higher loop order. The integral evaluates to
\begin{equation}
 \intsunset^{(1)}(p^2,\undernotation m^2)
  \;=\;
  \persunset^{(1)}(p^2,\undernotation m^2)
  \Bigl(-2R^{(1)}(p^2) + \log(4m_1^2 m_2^2)\Bigr),
  \label{eq:1loop}
\end{equation}
with holomorphic period
\begin{equation}
  \persunset^{(1)}(p^2,\undernotation m^2)
  \;=\;
  \frac{1}{\sqrt{(m_1^2+m_2^2-p^2)^2-4m_1^2m_2^2}},
  \label{eq:1loop-period}
\end{equation}
and flat coordinate
\begin{equation}
  R^{(1)}(p^2)
  \;=\;
  -\log\frac{m_1^2+m_2^2-p^2
    +\sqrt{(m_1^2+m_2^2-p^2)^2-4m_1^2m_2^2}}{2},
  \label{eq:1loop-R}
\end{equation}
satisfying
\begin{equation}
  \frac{dR^{(1)}}{dp^2} \;=\; \persunset^{(1)}(p^2,\undernotation m^2).
  \label{eq:1loop-mirror}
\end{equation}

The structure of equation~\eqref{eq:1loop} is the prototype for the all-loop result:
the integral is a product of the holomorphic period $\persunset^{(L)}$ and
an analytic function $R^{(1)}$.

\section{The all-equal-mass case}
\label{sec:equalmass}

The analysis of Section~\ref{sec:diffeq} singles out the all-equal-mass
configuration as the case where the Picard--Fuchs equation takes its
simplest form: a constant source term $-(L+1)!$ and a differential operator
of minimal order $L$. We now exploit this symmetry to derive a
closed-form solution for the sunset integral at all loop orders.

The central tool is the \emph{mirror map} and its associated
\emph{flat coordinate}
\begin{equation}
  R^{(L)}(p^2)
  \;=\;i \pi
  -\log(p^2)
  \;+\;
  \sum_{\substack{(r_1,\ldots,r_{L+1})\in\mathbb{N}^{L+1}\\
                  r_1+\cdots+r_{L+1}>0}}
  \left(\frac{(r_1+\cdots+r_{L+1})!}{r_1!\cdots r_{L+1}!}\right)^{\!2}
  \frac{1}{r_1+\cdots+r_{L+1}}
  \prod_{i=1}^{L+1}\frac{1}{(p^2)^{r_i}},
  \label{eq:flatcoord}
\end{equation}
such that its derivative gives the holomorphic period
\begin{equation}
  \frac{dR^{(L)}(p^2)}{dp^2}
  \;=\;
 \persunset^{(L)}(p^2,\mathbf{1})\,.
  \label{eq:flatcoord-def}
\end{equation}

Near $p^2=\infty$ we have
$R^{(L)}(p^2)\sim -\log(p^2)+O(1/p^2)$, so that $e^{R^{(L)}} \sim 1/p^2$
serves as a natural expansion parameter; the regime $|e^{R^{(L)}}|<1$
corresponds to large Euclidean momentum $|p^2|$ beyond the normal threshold.

The leading coefficient of the all-equal-mass differential operator is
given by~\cite{Vanhove:2014wqa} 
\begin{equation}
  c_L(p^2)
  \;:=\;
  \left.\mathscr{L}_L(p^2,\mathbf{1})\right|_{(d/dp^2)^L}
  \;=\;
  \begin{cases}
     (p^2)^{l}(p^2-1^2)(p^2-3^2)\cdots(p^2-(2l+1)^2)
     & L=2l, \\[4pt]
   (p^2)^{l+1}(p^2-2^2)(p^2-4^2)\cdots(p^2-(2l+2)^2)
      & L=2l+1.
  \end{cases}
  \label{eq:leading-coeff}
\end{equation}
For large $p^2$ one has $c_L(p^2)=(p^2)^{L+1}+\mathcal O((p^2)^{L})$.

\begin{proposition}
\label{prop:identity}
The holomorphic period and the flat coordinate satisfy the identity
\begin{equation}
  c_L(p^2)\,
  \bigl(\persunset^{(L)}(p^2,\mathbf{1})\bigr)^{L+1}
  \left(1 + \sum_{n\geq 1} (-1)^na^L(n)\,e^{nR^{(L)}(p^2)}\right)
  \;=\; 1,
  \label{eq:period-identity}
\end{equation}
where $a^L(n)\in\mathbb{Z}$ are integers. The first values for $L=2,3,4$
are listed in Table~\ref{tab:all_sequences}.
\end{proposition}

The integrality of the coefficients $a^L(n)$, together with their purely
exponential growth $(c_L\cdot(2L+2)!/L!^2)^n$ with small polynomial
corrections (visible in Table~\ref{tab:all_sequences}), is a hallmark of mirror
symmetry and will be explained in Remark~\ref{rmk:mirror}.

\subsection{Frobenius basis near \texorpdfstring{$p^2=\infty$}{p2=infty}}
\label{subsec:Frobenius}

\begin{definition}[Frobenius basis]
\label{def:Frobenius}
A Frobenius basis of solutions near $p^2=\infty$ is defined by
\begin{equation}
  \Frob_r^{(L)}(p^2)
  \;:=\;
  \frac{1}{r!}\,\persunset^{(L)}(p^2,\mathbf{1})\,
  \bigl(-i\pi+\log p^2\bigr)^r
  \;+\;
 \sum_{n\geq 1}{c_n\over (p^2)^n},
  \qquad 0\leq r\leq L-1,
  \label{eq:Frob-def}
\end{equation}
where $c_n$ are rational numbers.  We make the choice the principal
determination of the logarithm so that $\log(-p^2)=\log(p^2)-i\pi$
with the expansion performed in the region where $-p^2>0$.
\end{definition}

\medskip

In terms of the flat coordinate $R^{(L)}$, the Frobenius basis elements
admit the representation
\begin{equation}
  \Frob_r^{(L)}(p^2)
  \;=\;
  \frac{1}{r!}\,\persunset^{(L)}(p^2,\mathbf{1})\,
  \bigl(-R^{(L)}(p^2)\bigr)^r
  \;+\;
 \sum_{n\geq 1}\tilde{c}_n\,e^{nR^{(L)}(p^2)},
  \label{eq:Frob-R}
\end{equation}
where $\tilde{c}_n$ are rational numbers.

\medskip

\begin{remark}[Relation to the Frobenius basis of \cite{Bonisch:2021yfw}]
The Frobenius basis defined here coincides with the basis $g^L_i(p^2)$
used in~\cite{Bonisch:2021yfw} via $\Frob_r^{(L)} = g^L_r$.
The coefficients in that reference were determined numerically by analytic
continuation from below the normal threshold $p^2<(L+1)^2$; the present
paper provides them analytically to all orders via the Mellin representation.
\end{remark}

\subsection{The all-equal-mass sunset at all loop orders}
\label{subsec:main-theorem}

In this section we derive the expression of  the
all-equal-mass multiloop sunset integral in the using the mirror flat
coordinate $R^{(L)}(p^2)$.

\begin{theorem}[All-equal-mass case]
\label{thm:equalmass}
The all-equal-mass multiloop sunset integral is given by\footnote{An
  implementation of this expression is given in the worksheet \href{https://github.com/pierrevanhove/AllLoopSunset/blob/main/SingleMassCase-from-mellin.ipynb}{SingleMassCase-from-mellin.ipynb} in
  the repository~\cite{repoAllLoopSunset}. }
\begin{multline}
  \intsunset^{(L)}(p^2,\mathbf{1})
  \;=\;
 \persunset^{(L)}(p^2,\mathbf{1})
  \left(
    -(L+1)\bigl(-R^{(L)}(p^2)\bigr)^L
    + \sum_{l=1}^{\infty}
      e^{lR^{(L)}(p^2)}
      \sum_{r=0}^{L-1}
      d^{(L)}_r(l)\,\bigl(R^{(L)}(p^2)\bigr)^r
  \right)\cr
  + \sum_{r=0}^{L-1}
    \alpha_r^{(L)}\,\Frob_r^{(L)}(p^2),
  \label{eq:main-equalmass}
\end{multline}
where:
\begin{itemize}
  \item $d^{(L)}_r(l)\in\mathbb{Q}$ are rational coefficients
  \item $\alpha^{(1)}_0=\alpha_0^{(2)}=\alpha_1^{(2)}=0$ and for
    $L\geq 3$, $\alpha^{(L)}_r \in \mathbb{Q}[\zeta(3),\zeta(5),\zeta(7),\ldots]$
    are determined by the generating function
    \begin{equation}
      \label{e:Gdef}
      G(x,\lambda)=\sum_{n=0}^\infty\sum_{r=0}^{n}{\alpha^{(n+3)}_r
        \lambda^r x^n\over (n+4)!} =
      {2Q(x)\over 1+\lambda x}\, e^{x^3Q(x)}; \qquad
Q(x)=      \sum_{k=1}^\infty {\zeta(2k+1)\over2k+1}  x^{2k-2}    
    \end{equation}
  \item the series converges for $|e^{R^{(L)}(p^2)}|<1$, i.e.\ for
        $|p^2|$ sufficiently large beyond the normal threshold $(L+1)^2$.
\end{itemize}
\end{theorem}

\begin{proof}
In the all-equal-mass case $m_1=\cdots=m_{L+1}=1$, the general expansion of
Theorem~\ref{thm:generic} reduces to
\begin{equation}
  \intsunset^{(L)}(p^2,\mathbf{1})
  \;=\;
 \persunset^{(L)}(p^2,\mathbf{1})
  \left(
    -(L+1)\bigl(-\log p^2\bigr)^L
    + \sum_{a=0}^{L-1} \left(\sum_{n=1}^\infty c^{(L)}_a(n)/(p^2)^n\right)\,\bigl(-\log p^2\bigr)^a
  \right),
  \label{eq:proof-step1}
\end{equation}
where the
coefficients $c^{(L)}_a(n)$ are combinations of rational numbers and odd
zeta values as per the expansion of the ratio of the Gamma-functions
in the generic mass representation in equation~\eqref{e:IsunsetResult}. Changing variables from $p^2$ to the
flat coordinate via $1/p^2=\sum_{n\geq 1}\delta_n e^{nR^{(L)}(p^2)}$
and expressing $\log p^2=-R^{(L)}+O(e^{R^{(L)}})$ yields the
representation in equation~\eqref{eq:main-equalmass}, where we have
separated the dependence on the odd zeta values into the coefficients
of the Frobenius basis elements.

\medskip

The coefficients $\alpha_r^{(L)}$ are extracted by identifying the
$\log(p^2)$ terms.  Setting $\lambda=\log(p^2)$ in the single mass specialization of equation~\eqref{e:IsunsetResult} one
gets  that the coefficients of logarithm terms from the Frobenius
element at $L$ loop order is given by 
\begin{equation}
\alpha^{(L)}_r=(L+1)! \,
\operatorname{coeff}_{\lambda^{r}}\operatorname{coeff}_{x^{L-3}}\,\left(
  e^{-\gamma  x} \sqrt{\Gamma (1-x)\over \Gamma (x+1)} \frac{-2 \gamma  x+\log (\Gamma
    (1-x))-\log (\Gamma (1+x))}{x^3 (1+\lambda x)}\right) 
\end{equation}
where $\operatorname{coeff}_{x^{L-3}}$ is the coefficient of order
$x^{L-3}$ from the series expansion near $x=0$ and
$\operatorname{coeff}_{\lambda^{r}}$ is coefficient of $\lambda^r$.
Using that
\begin{equation}
   -2 \gamma  x+\log (\Gamma
    (1-x))-\log (\Gamma (1+x))=2\sum_{k=1}^\infty
    {\zeta(2k+1)\over2k+1} x^{2k+1}\,,
  \end{equation}
  we  obtain the result in equation~\eqref{e:Gdef}.

\end{proof}

\begin{remark}[Gamma class and the leading logarithm]
\label{rmk:gammaclass}
The leading term $-(L+1)(R^{(L)})^L$ in equation~\eqref{eq:main-equalmass} was
observed numerically in~\cite{Klemm:2019dbm, Bonisch:2020qmm} and shown to coincide
with the modified Gamma class of the degree-$(1,\ldots,1)$ Fano hypersurface
in $\mathbb{P}^{L+1}$. This has been established rigorously in~\cite{Iritani:2022}
(see also~\cite{Kerr:2022}).
\end{remark}

\medskip

\begin{remark}[Mirror map and instanton corrections]
\label{rmk:mirror}
The function $R^{(L)}(p^2)$ is a mirror map from the complex-structure
parameter $p^2$ to the K\"ahler parameter, generalising the one used
in~\cite{Bloch:2016izu}. Near $p^2=\infty$ one has
$R^{(L)}(p^2)\sim-\log(p^2)+O(1/p^2)$, so that $e^{R^{(L)}}\sim 1/p^2$
at large momentum. The exponential terms $e^{lR^{(L)}}$ in equation~\eqref{eq:main-equalmass}
are characteristic of instanton corrections in mirror symmetry, as illustrated in detail at two loops below.

Note that $R^{(L)}$ differs from the K\"ahler parameter $t_1$ defined from
the Frobenius basis at the point of maximal unipotent monodromy
\begin{equation}
 \mathfrak t_1
  \;=\;
  -\log(p^2)
  + \frac{1}{p^2}
  \sum_{(r_1,\ldots,r_{L+1})\in\mathbb{N}^{L+1}}
  \left(\frac{(r_1+\cdots+r_{L+1})!}{r_1!\cdots r_{L+1}!}\right)^{\!2}
  \left((L+1)\sum_{k=1}^{n}\frac{1}{k}-\sum_{i=1}^{L+1}\sum_{k=1}^{r_i}\frac{1}{k}\right)
  \prod_{i=1}^{L+1}\frac{1}{(p^2)^{r_i}},
  \label{eq:t1}
\end{equation}
by exponentially small corrections in $e^{R^{(L)}}$.
\end{remark}

\subsection{The two-loop sunset}
\label{subsec:2loop}

At two loops ($L=2$), with equal masses $m_1=m_2=m_3=1$, the general
formula in equation~\eqref{eq:main-equalmass} specialises to
\begin{equation}
  \intsunset^{(2)}(p^2,\mathbf{1})
  \;=\;
  \persunset^{(2)}(p^2)
  \left(
    -3\bigl(R^{(2)}(p^2)\bigr)^2
    + \sum_{l=1}^{\infty}
    \Bigl(d^{(2)}_0(l) + d^{(2)}_1(l)\,R^{(2)}(p^2)\Bigr)\,
    e^{lR^{(2)}(p^2)}
  \right).
  \label{eq:2loop}
\end{equation}
No Frobenius correction appears at $L=2$ since $L-r=2-0=2$ is even and odd
zeta values of even weight do not occur.

Using the code~\cite{repoAllLoopSunset}, we recover the results of~\cite{Bloch:2016izu}
\begin{equation}
  d^{(2)}_0(l) \;=\; -6l\,\mathrm{GW}(2,l),
  \qquad
  d^{(2)}_1(l) \;=\; 6l^2\,\mathrm{GW}(2,l), \qquad  \mathrm{GW}(2,l) \;=\; \sum_{d|l}\frac{1}{d^3}\,n_{l/d}
  \label{eq:2loop-coeff}
\end{equation}
where $n_l$ are the virtual genus-zero Gromov--Witten invariants of the
mirror Calabi--Yau threefold associated to the two-loop sunset elliptic
curve~\cite{Bloch:2016izu}. The first values for $1\leq l\leq 19$ are
\begin{align}
  \mathrm{GW}(2,l):&\quad
  1,\,-\tfrac{7}{8},\,\tfrac{28}{27},\,-\tfrac{135}{64},\,
  \tfrac{626}{125},\,-\tfrac{751}{54},\,\tfrac{14407}{343},\,
  -\tfrac{69767}{512},\,\tfrac{339013}{729},\,
  -\tfrac{827191}{500},\notag\\
  &\qquad
  \tfrac{8096474}{1331},\,-\tfrac{367837}{16},\,
  \tfrac{195328680}{2197},\,-\tfrac{137447647}{392},\,
  \tfrac{4746482528}{3375},\,-\tfrac{23447146631}{4096},\notag\\
  &\qquad
  \tfrac{115962310342}{4913},\,-\tfrac{574107546859}{5832},\,
  \tfrac{2844914597656}{6859},
  \label{eq:GWvalues}
\end{align}
and the virtual numbers $n_l$ for $1\leq l\leq 19$ are
\begin{align}
  n_l:&\quad
  1,\,-1,\,1,\,-2,\,5,\,-14,\,42,\,-136,\,465,\,-1655,\,6083,\,
  -22988,\notag\\
  &\qquad
  88907,\,-350637,\,1406365,\,-5724384,\,
  23603157,\,-98440995,\,414771045.
  \label{eq:virtualn}
\end{align}

\begin{remark}[Consistency check]
\label{rmk:check}
One verifies that equation~\eqref{eq:2loop} satisfies the all-equal-mass differential
equation
\begin{equation}
  \left[p^2(p^2-1)(p^2-9)\,\frac{d^2}{d(p^2)^2}
  + \bigl(3(p^2)^2-20p^2+9\bigr)\,\frac{d}{dp^2}
  + (p^2-3)\right]
  \intsunset^{(2)}(p^2)
  \;=\;
  -3!\,,
  \label{eq:2loop-DE}
\end{equation}
confirming Proposition~\ref{prop:source}. One also verifies the identity
\begin{equation}
 p^2(p^2-1)(p^2-9) \bigl(\persunset^{(2)}(p^2)\bigr)^3
  \left(1 -\sum_{n\geq 1} n^3\,\mathrm{GW}(2,n)\,
       e^{nR^{(2)}(p^2)}\right)
  \;=\;1,
  \label{eq:GW-identity}
\end{equation}
proven in~\cite{Bloch:2016izu}, which is the special case $L=2$
of Proposition~\ref{prop:identity}.
\end{remark}

\subsection{Higher-loop results and the Frobenius corrections}
\label{subsec:higher-loop}

At loop orders $L\geq 3$, the boundary coefficients $\alpha^{(L)}_r$ are
non-zero whenever $L-r$ is odd and positive. Writing the full integral as
the sum of an instanton series $\hat{I}^{(L)}_\circleddash$ (the
$e^{l R^{(L)}(p^2)}$
terms in equation~\eqref{eq:main-equalmass}) and the Frobenius correction, the
results for $L=3,\ldots,8$ are
\begin{align}
  \intsunset^{(3)}(p^2,\mathbf{1})
  &= \hat{I}^{(3)}_{\circleddash}(p^2,\mathbf{1})
     +16\,\zeta(3)\,\Frob^{(3)}_0(p^2),
  \notag\\
  \intsunset^{(4)}(p^2,\mathbf{1})
  &= \hat{I}^{(4)}_{\circleddash}(p^2,\mathbf{1})
     - 80\,\zeta(3)\,\Frob^{(4)}_1(p^2),
  \notag\\
  \intsunset^{(5)}(p^2,\mathbf{1})
  &= \hat{I}^{(5)}_{\circleddash}(p^2,\mathbf{1})
     + 288\,\zeta(5)\,\Frob^{(5)}_0(p^2)
     + 480\,\zeta(3)\,\Frob^{(5)}_2(p^2),
  \notag\\
  \intsunset^{(6)}(p^2,\mathbf{1})
  &= \hat{I}^{(6)}_{\circleddash}(p^2,\mathbf{1})
     + 1120\,\zeta(3)^2\,\Frob^{(6)}_0(p^2)
     - 2016\,\zeta(5)\,\Frob^{(6)}_1(p^2)
     - 3360\,\zeta(3)\,\Frob^{(6)}_3(p^2),
  \notag\\
  \intsunset^{(7)}(p^2,\mathbf{1})
  &= \hat{I}^{(7)}_{\circleddash}(p^2,\mathbf{1})
     + 11520\,\zeta(7)\,\Frob^{(7)}_0(p^2)
     - 8960\,\zeta(3)^2\,\Frob^{(7)}_1(p^2)
  \notag\\
  &\quad
     + 16128\,\zeta(5)\,\Frob^{(7)}_2(p^2)
     + 26880\,\zeta(3)\,\Frob^{(7)}_4(p^2),
  \notag\\
 \intsunset^{(8)}(p^2,\mathbf{1})
  &= \hat{I}^{(8)}_{\circleddash}(p^2,\mathbf{1})
     + 96768\,\zeta(5)\zeta(3)\,\Frob^{(8)}_0(p^2)
     - 103680\,\zeta(7)\,\Frob^{(8)}_1(p^2)
  \notag\\
  &\quad
     + 80640\,\zeta(3)^2\,\Frob^{(8)}_2(p^2)
     - 145152\,\zeta(5)\,\Frob^{(8)}_3(p^2)
     - 241920\,\zeta(3)\,\Frob^{(8)}_5(p^2),
  \label{eq:higher-loop}
\end{align}
where $\Frob^{(L)}_r(p^2)$ is defined
in equation~\eqref{eq:Frob-def} and has the behavior $\log(p^2)^r$ for large-$p^2$.
It can be checked that the coefficients in front of 
$\Frob_r^{(L)}(p^2)$ is given by the generating series in equation~\eqref{e:Gdef}.
The
coefficients for $L=3$ and $L=4$ are listed in full in
Appendices~\ref{sec:3loop1mass} and~\ref{sec:4loop1mass} respectively; all cases up
to $L=10$ are available in the repository~\cite{repoAllLoopSunset}.

\medskip

\begin{remark}
  The pattern of zeta values in equation~\eqref{eq:higher-loop} is consistent with
  the Gamma-class prediction of~\cite{Bonisch:2020qmm,Bonisch:2021yfw,Iritani:2022,Kerr:2022}: the
  boundary coefficients $\alpha^{(L)}_r$ are governed by the modified Gamma
  class $\hat\Gamma$ of the mirror Fano variety, whose expansion in terms of
  Riemann zeta values at odd arguments is well-known in mirror symmetry.
  Establishing this connection rigorously for $L\geq 3$ would constitute a
  natural extension of the proof in~\cite{Iritani:2022} and is left for
  future work.
\end{remark}

\begin{table}[h]
\centering
\begin{tabular}{r|r|r|r|r|r|r}
  \toprule
  \rowcolor{headerblue}
\multicolumn{2}{c|}{$L=2$} & \multicolumn{2}{c|}{$L=3$} & \multicolumn{2}{c}{$L=4$} \\
\cmidrule(lr){1-2} \cmidrule(lr){3-4} \cmidrule(lr){5-6}
$n$ & $a^2_n$ & $n$ & $a^3_n$ & $n$ & $a^4_n$ \\
\midrule
  1 & 1 & 1 & 4 & 1 & 10 \\
  \rowcolor{rowgray}
2 & 7 & 2 & 48 & 2 & 191 \\
  3 & 28 & 3 & 488 & 3 & 3,460 \\
  \rowcolor{rowgray}
4 & 135 & 4 & 5,296 & 4 & 65,161 \\
  5 & 626 & 5 & 57,228 & 5 & 1,234,080 \\
  \rowcolor{rowgray}
6 & 3,004 & 6 & 624,160 & 6 & 23,483,441 \\
  7 & 14,407 & 7 & 6,834,056 & 7 & 448,067,308 \\
  \rowcolor{rowgray}
8 & 69,767 & 8 & 75,080,688 & 8 & 8,565,472,861 \\
  9 & 339,013 & 9 & 826,829,636 & 9 & 163,965,975,856 \\
  \rowcolor{rowgray}
10 & 1,654,382 & 10 & 9,122,595,728 & 10 & 3,141,958,516,156 \\
  11 & 8,096,474 & 11 & 100,800,501,072 & 11 & 60,254,229,119,264 \\
  \rowcolor{rowgray}
12 & 39,726,396 & 12 & 1,115,131,619,744 & 12 & 1,156,222,352,587,939 \\
13 & 195,328,680 & 13 & 12,348,554,397,168 & 13 &
                                                  22,197,713,664,648,102 \\
  \rowcolor{rowgray}
14 & 962,133,529 & 14 & 136,855,381,207,200 & 14 & 426,332,197,232,125,896 \\
15 & 4,746,482,528 & 15 & 1,517,777,500,185,000 & 15 &
                                                       8,190,878,975,022,343,620 \\
  \rowcolor{rowgray}
16 & 23,447,146,631 & 16 & 16,842,669,239,950,832 & 16 & 157,409,777,121,794,035,317 \\
17 & 115,962,310,342 & 17 & 186,997,131,651,497,184 & 17 &
                                                           3,025,750,140,688,938,371,352 \\
  \rowcolor{rowgray}
18 & 574,107,546,859 & 18 & 2,077,070,943,532,153,648 & 18 & 58,172,760,670,776,391,232,102 \\
19 & 2,844,914,597,656 & 19 & 23,080,024,226,741,205,328 & 19 & 1,118,612,006,616,027,251,547,612 \\
\bottomrule
\end{tabular}
\caption{Sequence values for $L=2$, $L=3$, and $L=4$
  in equation~\eqref{eq:period-identity}. We notice that the
  coefficients have an exponential growth $(-c_L
(2L+2)!/L!^2)^n$ times small polynomial corrections and  with the
coefficients of order $c_L\simeq 1$.}
\label{tab:all_sequences}
\end{table}

\section{The sunset integral in four dimensions}\label{sec:finite4d}

Tarasov’s dimension-shifting relations~\cite{Tarasov:1996br,
  Tarasov:1997kx} connect Feynman integrals in different space–time
dimensions. By combining these relations with an expansion near  $D=4-2
\epsilon$, we obtain a systematic method for computing the pole and finite parts in four dimensions from results obtained in two dimensions.

\medskip

The key insight is that certain differential operators in the kinematic variables can shift the spacetime dimension. Specifically, for the sunset integral, applying an appropriate combination of derivatives with respect to $p^2$ to the $D$-dimensional integral yields the $(D+2)$-dimensional integral with modified propagator powers. While these relations are traditionally presented as descent relations (from $D+2$ to $D$), they can be reorganized into ascent relations~\cite{Lee:2009dh} relating the Feynman integral in dimension $D+2$ to the integral in dimension $D$ with shifted powers of the propagators. Consequently, any Feynman integral in $D+2$ dimensions can be expanded on a basis of master integrals in $D$ dimensions.

\medskip

Consider the all-equal-mass sunset integral normalized
by
\begin{equation}\label{e:normalised-int}
\nintsunset^{D,(L)}(p^2,\undernotation 1) =-
\frac{\intsunset^{D,(L)}(p^2,\undernotation 1)}{(L+1)!
  \left(\Gamma\left(4-D\over 2\right)\right)^L}\,.
\end{equation}
The normalized integral satisfies a minimal differential equation of order $L$~\cite{delaCruz:2024xit}:
\begin{equation}\label{e:PF1mass}
    \underbrace{\left( \sum_{r=0}^{L}q_r(p^2,D) \left(\frac{d}{dp^2}\right)^r
 \right)}_{=:\mathscr L_\circleddash^{D,(L)}}   \nintsunset^{D,(L)}(p^2,\undernotation
    1)=1\,,
\end{equation}
where the coefficients $q_r(p^2,D)$ are polynomials in $p^2$ of degree 
$2(L+1)-LD+r$ and polynomials in the dimension $D$. 

\medskip

\noindent\textbf{Basis of master integrals.} The general solution to equation~\eqref{e:PF1mass} is determined by $L$ linearly independent functions. A canonical choice of basis is given by successive derivatives:
\begin{equation}\label{e:Basis}
\mathcal{B}^{D,(L)} = \left\{\nintsunset^{D,(L)}, \frac{d}{dp^2}\nintsunset^{D,(L)}, \ldots, \left(\frac{d}{dp^2}\right)^{L-1}\nintsunset^{D,(L)}\right\}
\end{equation}
This basis consists of $L$ elements (derivatives up to order $L-1$) because the differential operator has order $L$. Higher derivatives can be expressed in terms of lower ones using equation~\eqref{e:PF1mass}.

\medskip
\noindent\textbf{Dimension-lowering operator.} The structure of this basis, combined with Tarasov's dimension-lowering relations, implies that the all-equal-mass sunset integral in $D+2$ dimensions can be obtained from the integral in $D$ dimensions by the action of a differential operator of order $L-1$:
\begin{equation}\label{eq:dimraise}
\nintsunset^{D+2,(L)} (p^2, \mathbf{1}) =
\underbrace{\left(\sum_{r=0}^{L-1} o^{(L)}_r(p^2, D)
    \left(\frac{d}{dp^2}\right)^r\right)}_{=:
  \mathscr{O}^{(L)}}\nintsunset^{D,(L)} (p^2, \mathbf{1})+ S^{(L)}(p^2,D)
\end{equation}
where $o^{(L)}_r(p^2, D)$ is a polynomial in $p^2$ of degree $L+r$ and
$S^{(L)}(p^2,D)$ is a polynomial in $p^2$ of order 
$L-1$ generated from the reduction of the reduced topologies that are
multiple tadpoles for the case of the sunset integrals.

\medskip

The order of the operator $\mathscr{O}^{(L)}$  is $L-1$ because the right-hand side must be expressible in the basis $\mathcal{B}^{D,(L)}$, which spans an $L$-dimensional space but consists of derivatives only up to order $L-1$.

\medskip
\noindent\textbf{System of differential equations.}
 Applying the differential operator $\mathscr{L}_\circleddash^{D+2,(L)}$ to equation~\eqref{eq:dimraise} and using equation~\eqref{e:PF1mass}, we obtain
\begin{equation}
\left(\mathscr{L}_\circleddash^{D+2,(L)} \circ
  \mathscr{O}^{(L)}\right) \nintsunset^{D,(L)}(p^2, \mathbf{1})
+\mathscr{L}_\circleddash^{D+2,(L)} S^{(L)}(p^2,D)
= 1\,.
\end{equation}
This  operator $\mathscr{L}_\circleddash^{D+2,(L)} \circ
  \mathscr{O}^{(L)}$  is reduced modulo $
  \mathscr{L}_\circleddash^{D,(L)}$ according
  \begin{equation}
        \mathscr{L}_\circleddash^{D+2,(L)} \circ
  \mathscr{O}^{(L)}=Q \circ\mathscr{L}_\circleddash^{D,(L)}+R
  \end{equation}
  and we use that  $
\nintsunset^{D,(L)}(p^2, \mathbf{1})$ satisfies the inhomogeneous
equation~\eqref{e:PF1mass} to get
\begin{equation}
Q(1)+R \nintsunset^{D,(L)}(p^2, \mathbf{1})
+\mathscr{L}_\circleddash^{D+2,(L)} S^{(L)}(p^2,D)
= 1\,.
\end{equation}
Since $R$ is a differential operator of order $L-1$ and the basis
$\mathcal{B}^{D,(L)} $ of master
integrals in equation~\eqref{e:Basis} is composed of independent integrals, we
must have that $R=0$, and $Q(1)
+\mathscr{L}_\circleddash^{D+2,(L)} S^{(L)}(p^2,d)
= 1$.
This gives a system of differential equations for the
coefficients $o^{(L)}_r(p^2, D)$ of $\mathscr{O}^{(L)}$ and $S^{(L)}(p^2,D)$.

\medskip

We determine $o^{(L)}_r(p^2, D)$ and
$S^{(L)}(p^2,D)$ by using that they are polynomials of known degree in $p^2$. 
This polynomial constraint is essential: it converts the differential
equations (which otherwise admit arbitrary integration constants) into
algebraic equations with unique solutions. The universality of
dimension-shifting relations requires that they be independent of
boundary conditions, which is achieved by this algebraic
determination.

\subsection{Dimension-lowering relations at one, two and three loops}

Using a dedicated \texttt{Maple} code~\cite{repoAllLoopSunset} and the known differential equations up to 20
loops determined in~\cite{delaCruz:2024xit} one can determine the
dimension-lowering operators  and then extract finite piece of the
multiloop sunset in four dimensions from derivatives of the two
dimensions result.
\medskip

We display the solutions at one-, two- and three-loops order.  Higher-order
results are easily obtained from the above mentioned code.

We obtain for the multiloop sunset integral
\begin{equation}\label{e:DimRaisingExplicit}
  \intsunset^{D+2\, (L)}(p^2,\undernotation 1)=
  \left(2\over 2-D\right)^L \,  \mathscr O^{(L)}
  \intsunset^{D\, (L)}(p^2,\undernotation 1)-(L+1)! \left(\Gamma\left(2-D\over 2\right)\right)^L\, S^{(L)}(p^2,D)\,.
\end{equation}

\noindent{\bf At one-loop order:}
The dimension-lowering operator is given by
\begin{equation}
  \mathscr O^{(1)}=-{D-2\over D-1}\,{p^2-4\over 4},  
\end{equation}
and the inhomogeneous piece by
\begin{equation}
  S^{(1)}(p^2,D)= -{1\over 2(D-1)}\,.  
\end{equation}

\noindent{\bf At two-loop order:}
The dimension-lowering operator is given by
\begin{multline}
  \mathscr O^{(2)}={(D-2)^2\over12 \left(D-1\right) \left(3 D-4\right)
    \left(3 D-2\right)}\Big[2 \left(p^2 -1\right) \left(p^2 +3\right)
  \left(p^2 -9\right) {d\over dp^2}\cr
 +(4 -D) (p^2)^{2}+(24-22 D) p^2 +87 D -156\Big]
  \end{multline}
  and the inhomogeneous part by
  \begin{equation}
    S^{(2)}(p^2,D)=\frac{\left(D-2\right) p^2 +21 D -30}{6 \left(D -1\right) \left(3 D -4\right) \left(3 D -2\right)}   \,.
  \end{equation}

  \noindent{\bf At three-loop order:}
  The dimension-lowering operator  is given by
  \begin{multline}
         \mathscr O^{(3)}={(D-2)^3\over 192 \left(D-1\right) \left(3
             D-4\right) \left(2 D-1\right) \left(2 D-3\right) \left(3
             D-2\right)}\cr
         \times\Big[-2 \left(p^2 -4\right) \left(p^2 -16\right)
         \left((p^2)^{2}+40 p^2 +64\right) p^2 \,\left(d\over dp^2\right)^{2}\cr
         +\bigg(\left(3 D-12\right) (p^2)^{4}+\left(200 D-500\right)
         (p^2)^{3}+\left(-1344 D+5376\right) (p^2)^{2}\cr
         -\left(D-3\right) \left(D-4\right) (p^2)^{3}+\left(-90
           D^{2}+542 D-736\right) (p^2)^{2}\cr
         +\left(-1392 D^{2}+3680 D-1648\right) p^2 \cr+4480 D^{2}-17792 D+18176
\Big]
  \end{multline}
  and the inhomogeneous part is given by
  \begin{equation}
    S^{(3)}(p^2,D)=\frac{-\left(D-2\right)^{2} (p^2)^{2}-8 \left(11 D-19\right) \left(D-2\right) p^2 -1216 D^{2}+3712 D-2848}{144 \left(D-1\right) \left(3 D-4\right) \left(2 D-1\right) \left(2 D-3\right) \left(3 D-2\right)}    \,.
  \end{equation}

  \medskip
  \begin{remark}\label{rem:epsilon}
As explained in~\cite{delaCruz:2024xit}, due to the absence of
apparent singularities, the dimension dependent all-equal-mass
differential operator for the multiloop sunset integral is given by
\begin{equation}
  \mathscr L_\circleddash^{4-2\epsilon,(L)}  = \mathscr
  L^{(L)}+\sum_{r=1}^{L} (D-2)^r \mathscr
  L^{(L-r)}
\end{equation}
where the differential operators $\mathscr
  L^{(r)} $ are of order $r$. This implies that the coefficient
  differential operator $\mathscr O^{(L)}$ are proportional to
  $(D-2)^L$, as
  \begin{equation}
  {\mathscr O}^{(L)}=  \sum_{r=0}^{L-1} o^{(L)}_r(p^2, D)
    \left(\frac{d}{dp^2}\right)^r= (D-2)^L  \sum_{r=0}^{L-1} \bar o^{(L)}_r(p^2, D)
    \left(\frac{d}{dp^2}\right)^r\,,
  \end{equation}
  where the coefficients $\bar o^{(L)}_r(p^2, D)$ have a finite
  non-vanishing limit for $D=2$.
 \end{remark}

  \subsection{Expansion near four dimensions}

The dimension-lowering relations given in the previous section take the
following form when applied to the case of $D=2-2\epsilon$,
\begin{equation}
    \intsunset^{D=4-2\epsilon\, (L)}(p^2,\undernotation 1)=  {1\over
    \epsilon^L} \, {\mathscr
    O}^{(L)}  \intsunset^{D=2-2\epsilon\, (L)}(p^2,\undernotation 1)
 -(L+1)!  \left(\Gamma\left(\epsilon\right)\right)^L \, S^{(L)}(p^2,2-2\epsilon)
\end{equation}
By plugging the expansion of the sunset near two dimension
\begin{equation}
  \intsunset^{D=2-2\epsilon\, (L)}(p^2,\undernotation 1)=e^{-L\gamma_E\epsilon}
  \jintsunset^{D=2\,(L)}(p^2,\undernotation 1)+O(\epsilon)\,,
\end{equation}
one can recover the expansion of the multiloop sunset near four dimensions
\begin{equation}
  \intsunset^{D=4-2\epsilon\, (L)}(p^2,\undernotation 1)=e^{-L\gamma_E\epsilon}\left(
  \sum_{r=1}^{L} {a_r^{(L)}(p^2)\over    \epsilon^r} +
  \jintsunset^{D=4\,(L)}(p^2,\undernotation 1)+O(\epsilon)\right)\,,
\end{equation}
where  $\jintsunset^{D=4\,(L)}(p^2,\undernotation 1)$ is the finite
part of the sunset integral in four dimensions.

\medskip
Thanks to  the factorisation of $\epsilon^L$ from the differential
operator $\mathscr O^{(L)}$, as explained in remark~\ref{rem:epsilon},
and 
$\epsilon\Gamma(\epsilon)=\Gamma(1+\epsilon)$ the differential part is
finite for $\epsilon\to0$ but the inhomogeneous term has poles in
$\epsilon$.

\medskip

In particular, one can obtain the finite part in four dimensions from
derivative of the finite part in two dimensions, with the relation
that takes the form 
\begin{equation}\label{e:Result4d}
   \jintsunset^{D=4-2\epsilon\, (L)}(p^2,\undernotation
   1)=\left(\sum_{r=0}^{L-1} \bar c^{(L)}_r(p^2) \left(d\over dp^2\right)^r\right) \jintsunset^{D=2\, (L)}(p^2,\undernotation 1)+R_4^{(L)}(p^2)\,,
\end{equation}
where the coefficients $\bar
c^{(L)}_r(p^2)=\lim_{\epsilon\to0}
\epsilon^{-L} c_r^{(L)}(p,2-2\epsilon)$
where $c_r^{(L)}(p,2-2\epsilon)$ are the coefficients of ${\mathscr
  O}^{(L)}$ in equation~\eqref{e:DimRaisingExplicit} and the 
inhomogeneous term reads
\begin{equation}
-(L+1)!\left(\Gamma(\epsilon)\right)^L
e^{L\epsilon\gamma} S^{(L)}(p^2,2-2\epsilon)=\sum_{r=1}^L
{a_r^{(L)}(p^2)\over\epsilon^r} + R_4^{(L)}(p^2)+\mathcal O(\epsilon)\,.
\end{equation}

\medskip
We list the results at one-, two- and three-loop order,more examples are available in the repository~\cite{repoAllLoopSunset}.

\smallskip

\noindent{\bf At one-loop order:} For $L=1$, we have the relation
between the four dimensional finite piece and the two-dimensional
value

\begin{equation}
  \intsunset^{D=4\,(1)}(p^2,\undernotation 1)= {1\over \epsilon}+ {p^2-4\over4}   \intsunset^{D=2,(1)}(p^2,\undernotation 1)+{1\over2}\,.
\end{equation}

\noindent{\bf At two-loop order:}
 For $L=2$, we have the explicit relation between the finite piece in four dimensions and the value of the integral in two dimensions~\cite{Laporta:2004rb,Bloch:2013tra}
 \begin{multline}\label{eq:2loopshift}
        \intsunset^{D=4\,(2)}(p^2,\undernotation 1^2)=-{3\over
          2\epsilon^2} + {p^2-18 \over 4\epsilon}\cr
+  {(p^2-1)(p^2-9)\over 12}\left(1+(p^2+3){d\over dt}\right)
  \intsunset^{D=2,(2)}(p^2,\undernotation 1)+ {13p^2-72\over128}-{\pi^2\over4}
\end{multline}

\noindent{\bf At three-loop order:}
For $L=3$, we recover the relation derived in~\cite{Lellouch:2025rnz}
(taking into account that $E_1(2;p^2)= - \intsunset^{D=2,(3)}(p^2,\undernotation 1)$)

\begin{multline}\label{eq:3loopshift}
   \intsunset^{D=4\,(3)}(p^2,\undernotation 1^2)
   ={2\over\epsilon^3} + {23 - p^2\over3\epsilon^2} + {630 + 18\pi^2 + p^2(p^2-54)\over36\epsilon}\cr
+\frac{p^2 \left(p^2 -4\right) \left(p^2 -16\right) \left((p^2)^{2}+40 p^2 +64\right)}{288}
   \left(d\over dp^2\right)^2  \intsunset^{D=2,(3)}(p^2,\undernotation 1)\cr
   +{3 (p^2)^{4}+50 (p^2)^{3}-1344 (p^2)^{2}+1920 p^2 +4096\over 288} {d\over dp^2} \intsunset^{D=2,(3)}(p^2,\undernotation 1)\cr
   +{(p^2)^{3}+6 (p^2)^{2}-72 p^2 -256\over 288} \intsunset^{D=2,(3)}(p^2,\undernotation 1)\cr
        + \frac{71 (p^2)^2+18p^2+6102}{216} 
        + \frac{(23-p^2)\pi^2}{12} - 2\zeta(3)\,.
      \end{multline}

Higher orders in the $\epsilon$ expansion require knowledge of the corresponding higher-order $\epsilon$ terms in $D = 2 - 2\epsilon$. The $\epsilon$-expansion in two dimensions, derived in~\cite{Duhr:2025ouy}, is sufficient to determine the expansion in four dimensions through the application of the dimension-lowering formula.
      
\appendix
\section{The three-loop case}\label{sec:3loop1mass}
The three-loop all-equal-mass integral is
\begin{multline}\label{e:IQ3loop}
  \intsunset^{(3)}(p^2,\undernotation 1)= \persunset^{(3)}(p^2,\undernotation 1)  \left(4 (R^{(3)}(p^2))^3
    +\sum_{l=1}^\infty e^{lR^{(3)}(p^2)} \sum_{r=0}^{2}
    d^{(3)}_r (l)(R^{(3)}(p^2))^r\right)\cr
  +16\zeta(3)\Frob^{0\,(3)}(p^2)\,.
\end{multline}
The coefficients $d_r^{(3)}(l)$ with $0\leq
r\leq 2$ are given in the table~\ref{tab:3loop-coeffs}. These
coefficients are obtained using the code in~\cite{repoAllLoopSunset}.

\captionsetup{width=15cm}
\begin{landscape}
 \begin{center}
{\large\bfseries Coefficients $d_0^{(3)}(\ell)$, $d_1^{(3)}(\ell)$ and $d_2^{(3)}(\ell)$ for the $L=3$ Multiloop Sunset}
\end{center}
\vspace{4pt}

\begin{longtable}{%
  C{0.5cm}  
  R{5.8cm}  
  R{5.8cm}  
  R{5.8cm}  
  }
\toprule
\rowcolor{headerblue}
$\ell$ & $(-1)^{\ell+1}d_0^{(3)}(\ell)$ & $(-1)^\ell d_1^{(3)}(\ell)$ &
                                                                $(-1)^\ell d_2^{(3)}(\ell)$ \\
\midrule
\endfirsthead

\toprule
\rowcolor{headerblue}
$\ell$ & $(-1)^{\ell+1}d_0^{(3)}(\ell)$ & $(-1)^\ell d_1^{(3)}(\ell)$ &
                                                                $(-1)^\ell d_2^{(3)}(\ell)$ \\
\midrule
\endhead

\bottomrule
\endfoot

\rowcolor{white}
$1$ &
$48$ &
$0$ &
$ 24$ \\

\rowcolor{rowgray}
$2$ &
$ 138$ &
$48$ &
$132$ \\

\rowcolor{white}
$3$ &
$\dfrac{7144}{9}$ &
$ 528$ &
$ 800$ \\[8pt]

\rowcolor{rowgray}
$4$ &
$ \dfrac{63695}{12}$ &
$4652$ &
$6018$ \\[8pt]

\rowcolor{white}
$5$ &
$\dfrac{15112894}{375}$ &
$ 41672$ &
$ \dfrac{242424}{5}$ \\[8pt]

\rowcolor{rowgray}
$6$ &
$ \dfrac{73656439}{225}$ &
$\dfrac{5695028}{15}$ &
$415088$ \\[8pt]

\rowcolor{white}
$7$ &
$\dfrac{23997011104}{8575}$ &
$ \dfrac{17647088}{5}$ &
$ \dfrac{25875552}{7}$ \\[8pt]

\rowcolor{rowgray}
$8$ &
$ \dfrac{116790128543}{4704}$ &
$\dfrac{233804429}{7}$ &
$33914049$ \\[8pt]

\rowcolor{white}
$9$ &
$\dfrac{134933770848983}{595350}$ &
$ \dfrac{33700241452}{105}$ &
$ \dfrac{954894248}{3}$ \\[8pt]

\rowcolor{rowgray}
$10$ &
$ \dfrac{1399668884020621}{661500}$ &
$\dfrac{1640310165758}{525}$ &
$\dfrac{15208750932}{5}$ \\[8pt]

\rowcolor{white}
$11$ &
$\dfrac{885466604960655358}{44022825}$ &
$ \dfrac{645818708336}{21}$ &
$ \dfrac{324462499968}{11}$ \\[8pt]

\rowcolor{rowgray}
$12$ &
$ \dfrac{2662835161020746401}{13721400}$ &
$\dfrac{151266708995597}{495}$ &
$289525458104$ \\[8pt]

\rowcolor{white}
$13$ &
$\dfrac{3030691644857404805897}{1598647050}$ &
$ \dfrac{321470988355756}{105}$ &
$ \dfrac{37323117924480}{13}$ \\[8pt]

\rowcolor{rowgray}
$14$ &
$ \dfrac{29535849328844384553856}{1578151575}$ &
$\dfrac{216476734691524220}{7007}$ &
$\dfrac{201031033602384}{7}$ \\[8pt]

\rowcolor{white}
$15$ &
$\dfrac{1891441425977079362775646}{10145260125}$ &
$ \dfrac{10093926006696084736}{32175}$ &
$ 289441832285280$ \\[8pt]

\rowcolor{rowgray}
$16$ &
$ \dfrac{64817896133977514308038803}{34629154560}$ &
$\dfrac{192402703062289059181}{60060}$ &
$\dfrac{5872668577670337}{2}$ \\[8pt]

\rowcolor{white}
$17$ &
{\small$\dfrac{71854700832696435940625956553}{3797612418600}$} &
$ \dfrac{493620493470363121754}{15015}$ &
$ \dfrac{509335737033576096}{17}$ \\[8pt]

\rowcolor{rowgray}
$18$ &
{\footnotesize$ \dfrac{416621283168424043953343380081}{2165154591600}$} &
{\small$\dfrac{59884860269237019388619}{176715}$} &
$\dfrac{921823061292056348}{3}$ \\[8pt]

\rowcolor{white}
$19$ &
{\scriptsize$\dfrac{5275263677619248981107022533667971}{2681393603738850}$} &
{\small$ \dfrac{59682200835264444248792}{17017}$} &
$ \dfrac{60149764030617173760}{19}$ \\[8pt]

  \bottomrule
\caption{Coefficients for the three-loop all-equal-mass expansion in equation~\eqref{e:IQ3loop}.}
\label{tab:3loop-coeffs}
\end{longtable}
\end{landscape}
\section{The four-loop case}\label{sec:4loop1mass}
The four-loop all-equal-mass integral is 
\begin{multline}\label{eq:4loop-appendix}
\intsunset^{(4)}(p^2, \mathbf{1}) = \persunset^{(4)}(p^2, \mathbf{1})
\left( -5(R^{(4)}(p^2))^4 + \sum_{l=1}^\infty e^{l R^{(4)}(p^2)}
  \sum_{r=0}^{3} d_r^{(4)}(l) (R^{(4)}(p^2))^r \right)\cr
-80\zeta(3) \Frob^{1\,(4)}(p^2)\,.
\end{multline}
The coefficients 
$d_r^{(4)}(l)$, with $0\leq r\leq 3$ are given in
tables~\ref{tab:4loop-c0-c1-coeffs} and~\ref{tab:4loop-c2-c3-coeffs}. These
coefficients are obtained using the code in~\cite{repoAllLoopSunset}.
\begin{landscape}

\begin{center}
{\large\bfseries Coefficients $d_0^{(4)}(\ell)$ and $ d_1^{(4)}(\ell)$ for the $L=4$ Multiloop Sunset}
\end{center}
\vspace{4pt}

\begin{longtable}{C{0.4cm} C{9.0cm} C{9.0cm}}
\toprule
\rowcolor{headerblue}
$\ell$ & $(-1)^\ell d_0^{(4)}(\ell)$ & $(-1)^\ell d_1^{(4)}(\ell)$ \\
\midrule
\endfirsthead
\toprule
\rowcolor{headerblue}
$\ell$ & $(-1)^\ell d_0^{(4)}(\ell)$ & $(-1)^\ell d_1^{(4)}(\ell)$ \\
\midrule
\endhead
\bottomrule
\endfoot

\rowcolor{white}  $1$  & $0$ & $240$ \\[2pt]
\rowcolor{rowgray}$2$  & $630$ & $870$ \\[2pt]
\rowcolor{white}  $3$  & $\F{23620}{3}$ & $\F{53450}{9}$ \\[8pt]
\rowcolor{rowgray}$4$  & $\F{2328505}{24}$ & $\F{187075}{4}$ \\[8pt]
\rowcolor{white}  $5$  & $\F{5166121}{4}$ & $\F{20821971}{50}$ \\[8pt]
\rowcolor{rowgray}$6$  & $\F{16311974269}{900}$ & $\F{694521607}{180}$ \\[8pt]
\rowcolor{white}  $7$  & $\F{111336261833}{420}$ & $\F{23978615637}{686}$ \\[8pt]
\rowcolor{rowgray}$8$  & $\F{13184469178409351}{3292800}$ & $\F{2159501115167}{7840}$ \\[8pt]
\rowcolor{white}  $9$  & $\Fs{4138135621730876579}{66679200}$ & $\F{532452195983189}{476280}$ \\[8pt]
\rowcolor{rowgray}$10$ & $\Fs{218340339433172908237}{222264000}$ & $ \F{4143926311480331}{176400}$ \\[8pt]
\rowcolor{white}  $11$ & $\Fs{1289503700192366915033}{81496800}$ & $ \Fs{7332249548528853169}{7826280}$ \\[8pt]
\rowcolor{rowgray}$12$ & $\Fs{9180926381857951580193311}{35500006080}$ & $ \Fs{118424516949509090311}{5122656}$ \\[8pt]
\rowcolor{white}  $13$ & $\Ff{51438346783674293529672868}{12018231225}$ & $ \Fs{289829501049924675872932}{586170585}$ \\[8pt]
\rowcolor{rowgray}$14$ & $\Ff{9306807772205606379343991787551}{129989188929600}$ & $ \Fs{10008094049987378558822845}{1010017008}$ \\[8pt]
\rowcolor{white}  $15$ & $\Fsc{32738328734448034104605689186103}{27081081027000}$ & $ \Fsc{2881744581783241678853344}{15030015}$ \\[8pt]
\rowcolor{rowgray}$16$ & $\Fsc{9511608121403402401605752312380891}{462183782860800}$ & $ \Ff{932186741645270934232073869}{256512256}$ \\[8pt]
\rowcolor{white}  $17$ & $\Fsc{207946720728469354554430199795312981}{589284323147520}$ & $ \Fsc{32137264368294370054091612587807}{472591767648}$ \\[8pt]
\rowcolor{rowgray}$18$ & $\Fsc{25456747345278724838046021646302868916011}{4180168703200089600}$ & $ \Fsc{2583563035969220472370362290172271}{2047055250240}$ \\[8pt]
\rowcolor{white}  $19$ & $\Fsc{1026102856392132733834308599203930637688581}{9707280655209096960}$ & $ \Fsc{22213605712201435099889872621957081277}{953384392440480}$ \\[8pt]

  \bottomrule
   \caption{Coefficients $ d_0^{(4)}(\ell)$ and $d_1^{(4)}(\ell)$ for the
  four-loop all-equal-mass expansion in equation~\eqref{eq:4loop-appendix}.}
\label{tab:4loop-c0-c1-coeffs}
\end{longtable}

\newpage

\begin{center}
{\large\bfseries Coefficients $ d_2^{(4)}(\ell)$ and $ d_3^{(4)}(\ell)$ for the $L=4$ Multiloop Sunset}
\end{center}
\vspace{4pt}

\begin{longtable}{C{0.4cm} C{9.0cm} C{9.0cm}}
\toprule
\rowcolor{headerblue}
$\ell$ & $(-1)^{\ell+1}d_2^{(4)}(\ell)$ & $(-1)^{\ell+1}d_3^{(4)}(\ell)$ \\
\midrule
\endfirsthead
\toprule
\rowcolor{headerblue}
$\ell$ & $(-1)^{\ell+1} d_2^{(4)}(\ell)$ & $(-1)^{\ell+1} d_3^{(4)}(\ell)$ \\
\midrule
\endhead
\bottomrule
\endfoot

\rowcolor{white}  $1$  & $ 60$ & $ 60$ \\[2pt]
\rowcolor{rowgray}$2$  & $ 705$ & $ 450$ \\[2pt]
\rowcolor{white}  $3$  & $ \F{26660}{3}$ & $ 4220$ \\[8pt]
\rowcolor{rowgray}$4$  & $ \F{455745}{4}$ & $ 50145$ \\[8pt]
\rowcolor{white}  $5$  & $ \F{7846812}{5}$ & $ 656112$ \\[8pt]
\rowcolor{rowgray}$6$  & $ \F{68065687}{3}$ & $ 9175170$ \\[8pt]
\rowcolor{white}  $7$  & $ \F{16673722020}{49}$ & $ \F{940420800}{7}$ \\[8pt]
\rowcolor{rowgray}$8$  & $ \F{588043043271}{112}$ & $ \F{4072262625}{2}$ \\[8pt]
\rowcolor{white}  $9$  & $ \F{15661912577882}{189}$ & $ \F{95081915240}{3}$ \\[8pt]
\rowcolor{rowgray}$10$ & $ \F{46629859339087}{35}$ & $ 503956078440$ \\[8pt]
\rowcolor{white}  $11$ & $ \Fs{18422720261265234}{847}$ & $ \F{89690700456780}{11}$ \\[8pt]
\rowcolor{rowgray}$12$ & $ \Fs{332379263011895915}{924}$ & $ 133842621289025$ \\[8pt]
\rowcolor{white}  $13$ & $ \Fs{78276381566852736844}{13013}$ & $ \Fs{28912327783023420}{13}$ \\[8pt]
\rowcolor{rowgray}$14$ & $ \Fs{711581788779235457945}{7007}$ & $ \Fs{261412950792522720}{7}$ \\[8pt]
\rowcolor{white}  $15$ & $ \Fs{8652798193114594186344}{5005}$ & $ 632782893995694844$ \\[8pt]
\rowcolor{rowgray}$16$ & $ \Fs{1899387916957131801992745}{64064}$ & $ \Fs{43229858131865045985}{4}$ \\[8pt]
\rowcolor{white}  $17$ & $ \Ff{148049273275274835073944645}{289289}$ & $ \Fs{3159947780099952023400}{17}$ \\[8pt]
\rowcolor{rowgray}$18$ & $ \Ff{742264599209072634928184983}{83538}$ & $ 3216909271990331115460$ \\[8pt]
\rowcolor{white}  $19$ & $ \Ff{952652961057018330302763844447}{6143137}$ & $ \Fs{1063694718739578260279520}{19}$ \\[8pt]

  \bottomrule
  \caption{Coefficients $d_2^{(4)}(\ell)$ and $d_3^{(4)}(\ell)$ for the
  four-loop all-equal-mass expansion in equation~\eqref{eq:4loop-appendix}.}
\label{tab:4loop-c2-c3-coeffs}
\end{longtable}

\end{landscape}


\end{document}